\documentclass{amsproc}

\usepackage{amssymb}

\newtheorem{theorem}{Theorem}[section]
\newtheorem{lemma}[theorem]{Lemma}
\newtheorem{prop}[theorem]{Proposition}
\newtheorem{coro}[theorem]{Corollary}

\theoremstyle{remark}
\newtheorem{remark}[theorem]{Remark}

\newcommand{\N} {\ensuremath{\mathbb{N}}}
\newcommand{\F} {\ensuremath{\mathbb{F}}}

\newcommand{\G} {\textnormal{\textrm{G}}}
\newcommand{\PP} {\textnormal{\textrm{P}}}
\newcommand{\Aut} {\textnormal{\textrm{Aut}}}
\newcommand{\VV} {\mathcal{V}}
\newcommand{\CC} {\mathcal{C}}
\newcommand{\EE} {\mathcal{E}}
\newcommand{\id} {\textnormal{\textrm{id}}}

\newcommand{\soc} {\textnormal{\textrm{soc}}}
\newcommand{\sss} {\textnormal{\textrm{s}}}

\numberwithin{equation}{section}

\begin{document}

\title[Automorphism groups of binary linear codes]{On the automorphism groups of binary linear codes}

\author{Martino Borello}
\address{Dipartimento di Matematica e Applicazioni, Universit\`{a} degli Studi di Milano - Bi\-coc\-ca, via Cozzi 55, 20125 Milano}
\email{m.borello1@campus.unimib.it}
\thanks{Member INdAM-GNSAGA (Italy), IEEE}


\date{}

\begin{abstract}
Let $\CC$ be a binary linear code and suppose that its automorphism
group contains a non trivial subgroup $G$. What can we say about
$\CC$ knowing $G$? In this paper we collect some answers to this
question in the cases $G\cong C_p$, $G\cong C_{2p}$ and $G\cong
D_{2p}$ ($p$ an odd prime), with a particular regard to the case in
which $\CC$ is self-dual. Furthermore we generalize some methods
used in other papers on this subject. Finally we give a short survey
on the problem of determining the automorphism group of a putative
self-dual $[72,36,16]$ code, in order to show where these methods
can be applied.
\end{abstract}

\maketitle

This paper is a presentation of some of the main results about the
automorphism group of binary linear codes obtained by the author in
his Ph.D. thesis. Part of the results are proved in joint papers
with Wolfgang Willems, Francesca Dalla Volta and Gabriele Nebe.

The problem we want to investigate is the following: let $\CC$ be a
(self-dual) binary linear code and suppose that $\Aut(\CC)$ contains
a non trivial subgroup $G$. What can we say about $\CC$ knowing $G$?

To face this problem, usually we want to find out ``smaller pieces''
which are easier to determine and then look at the structure of the
whole code.

In \mbox{Section \ref{sectionprime}} we present a classical
decomposition of codes with automorphisms of odd prime order. In
\mbox{Section \ref{section2p}} we summarize the most significant
results of \cite{BW}, about codes with automorphisms of order $2p$,
where $p$ is an odd prime. \mbox{Section \ref{sectiondihedral}} is a
ge\-ne\-ra\-lization of methods used in \cite{FeulNe} and
\cite{BDN}, about codes whose automorphism groups contain particular
dihedral groups. Finally, in \mbox{Section \ref{sectioninter}} we
point out and generalize some theoretical tools used in
\cite{Baut6}, \cite{BDN} and \cite{BE8}.\\
Our methods can be applied
\begin{itemize}
  \item to study the possible automorphism groups of extremal self-dual binary linear codes;
  \item to construct self-orthogonal binary linear codes with large minimum distan\-ce and relatively large dimension;
  \item to classify self-dual binary linear codes with certain parameters.
\end{itemize}
Obviously the last one is the most ambitious.

In the last section, which is a short survey on the problem of
determining the automorphism group of a putative extremal self-dual
$[72,36,16]$ code, we underline where these methods can be applied,
showing their power.

\section{Background and notations}

We refer the reader to \cite{HP} for basic notions of Coding Theory
and to \cite{HB} for basic notions of Representation Theory. In this
section we just want to fix the notations we use.

Let $\CC$ be an $[n,k,d]$ code. Then we denote by $\G(\CC)$ a
\emph{generator matrix} of $\CC$, i.e. a matrix in
$\text{Mat}_{k,n}(\F_2)$ whose rows generate $\CC$.

\noindent Let $\sigma\in S_n$. Then we define
$\CC^\sigma:=\{c^\sigma \ | \ c\in\CC\}$. The \emph{automorphism
group} of $\CC$ is
$$\Aut(\CC):=\{\sigma\in S_n \ | \ \CC^\sigma=\CC\}\leq S_n.$$
The \emph{fixed code of} $\sigma$ is defined as
$$\CC(\sigma):=\{c\in\CC \ | \ c^\sigma=c\},$$
that is obviously a subcode of $\CC$.\\ If we call
$\Omega_1,\ldots,\Omega_{m_\sigma}$ the orbits of $\sigma$ on the
coordinates $\{1,\ldots,n\}$, we have trivially that
$c=(c_1,\ldots,c_n)\in\CC$ is in $\CC(\sigma)$ if and only if
$c_i=c_j$ for all $i,j\in\Omega_k$, for every
$k\in\{1,\ldots,m_\sigma\}$. In this case we say that $c$ is
\emph{constant on the orbits of} $\sigma$. Thus we can define a
\emph{natural projection associated to} $\sigma$
$$\pi_\sigma:\CC(\sigma)\rightarrow \F_2^{m_\sigma}$$
such that $(\pi_\sigma(c))_k:=c_h$ for any $h\in \Omega_k$, which is
clearly well-defined for $c\in\CC(\sigma)$. \\
If $\sigma$ is a permutation of order $p$ we say that $\sigma$ is of
\emph{type} $p$-$(c,f)$ if it has $c$ cycles of length $p$ and $f$
fixed points. \\
If $\sigma$ is a permutation of order $p\cdot q$ we say that
$\sigma$ is of \emph{type} $p\cdot q$-$(a,b,c;f)$ if it has $a$
cycles of length $p$, $b$ cycles of length $q$, $c$ cycles of length
$p\cdot q$ and $f$ fixed points.

Let $\mathcal{C},\mathcal{D}\leq\F_2^n$. We set
$\CC+\mathcal{D}:=\{c+d \ | \ c\in\CC, d\in\mathcal{D}\}$, sum of
$\CC$ and $\mathcal{D}$. If $\CC\cap\mathcal{D}=\{\textbf{0}\}$, we
say that the sum is \emph{direct} and we denote it by
$\CC\oplus\mathcal{D}$. This should not be confused with another
common concept of direct sum of codes, which we do not use in this
paper.

We use the following notations for groups:
\begin{itemize}
\item $C_n$ is the \emph{cyclic group} of order $n$;
\item $D_n$ is the \emph{dihedral group} of order $n$;
\item $S_n$ is the \emph{symmetric group} of degree $n$;
\item $A_n$ is the \emph{alternating group} of degree $n$.
\end{itemize}
Furthermore, for $H,G$ groups, $H\times G$ is the \emph{direct
product} of $H$ and $G$ while $H\rtimes G$ is a \emph{semidirect
product} of $H$ and $G$. If $H\leq G$, we denote the
\emph{centralizer} and the \emph{normalizer} of $H$ in $G$ by
$\textnormal{C}_G(H)$ and $\textnormal{N}_G(H)$ respectively.

We conclude giving the definition of a fundamental number: we denote
by $\sss(p)$ the \emph{multiplicative order of} $2$ in
$\F_p^\times$, i.e. the smallest $m\in\N$ such that $p \ | \ 2^m-1$.

\section{Cyclic group of order $p$ ($p$ an odd prime)}\label{sectionprime}

In this section we introduce a well-known classical decomposition of
codes with automorphisms of odd prime order. We want to present it
for completeness, although it is just a particular reformulation of
Maschke's Theorem, and to fix some notations useful in the
following.

Let $\VV:=\F_2^n$ and $\sigma\in S_n$ a permutation of odd prime
order $p$. Then, it is trivial to prove that
$$\VV=\VV(\sigma)\oplus \VV(\sigma)^\perp$$
where $\VV(\sigma)$ is the subspace fixed by $\sigma$ and
$\VV(\sigma)^\perp$ is the dual of $\VV(\sigma)$, that is clearly
the subspace of even-weight vectors on the orbits of $\sigma$. We
note that $\CC(\sigma)=\CC\cap\VV(\sigma)$ and we define
$\EE(\sigma):=\CC\cap \VV(\sigma)^\perp$. Then we have the
following.
\begin{theorem}\label{theoremdecompCE}
Let $\CC$ be a binary linear code and suppose $\sigma\in\Aut(\CC)$
of odd prime order $p$. Then $$\CC=\CC(\sigma)\oplus \EE(\sigma),$$
where $\CC(\sigma)$ is the fixed code of $\sigma$ and $\EE(\sigma)$
is the subcode of even-weight codewords on the orbits of $\sigma$.
\end{theorem}

In order to get more information on the subcode $\EE(\sigma)$, with
a particular regard to the case in which $\CC$ is self-dual, we
investigate more closely the decomposition of $\VV$ . Firstly we
consider the case in which $n=p$ and then the general case.

\vspace{2mm} Let $n=p$, so that $\sigma$ is of type $p$-$(1,0)$.
Thus
$$\textnormal{G}(\VV(\sigma))=\left[
  \begin{smallmatrix}
    1 & 1 & 1 & 1 & \ldots & 1 & 1 \\
  \end{smallmatrix}
\right] \ \  \text{and} \ \ \textnormal{G}(\VV(\sigma)^\perp)=\left[
  \begin{smallmatrix}
        &   &   &   &        &   &   \\
      &   &   &   &        &   &   \\
    1 & 1 & 0 & 0 & \ldots & 0 & 0 \\
    0 & 1 & 1 & 0 & \ldots & 0 & 0 \\
    0 & 0 & 1 & 1 & \ldots & 0 & 0 \\
    \vdots & \vdots & \vdots & \vdots & \ddots & \vdots & \vdots \\
    0 & 0 & 0 & 0 & \ldots & 1 & 1 \\
          &   &   &   &        &   &   \\
  \end{smallmatrix}
\right]$$ in $\text{Mat}_{1,p}(\F_2)$ and $\text{Mat}_{p-1,p}(\F_2)$
respectively.\\
There is a natural isomorphism of vector spaces
\begin{equation}\label{eqphi}
\varphi:\F_2^p\rightarrow \F_2[x]/(x^p+1)=:Q
\end{equation}
which maps $(v_0,\ldots,v_{p-1})\mapsto
v_0+\ldots+v_{p-1}x^{p-1}$.\\
Notice that
$(x^p+1)=(x+1)(x^{p-1}+x^{p-2}+\ldots+x+1)$, with $(x+1)$ and
$(x^{p-1}+x^{p-2}+\ldots+x+1)$ coprime (since $p$ is odd). It is
well-known that the polynomial $(x^{p-1}+x^{p-2}+\ldots+x+1)$ is the
product of $t:=\frac{p-1}{\sss(p)}$ irreducible polynomials of
degree $\sss(p)$. So, let $(x^p+1)=q_0(x) q_1(x)\ldots q_t(x)$,
where $q_0(x):=(x+1)$ and the other terms are the $t$ irreducible
polynomials of degree $\sss(p)$. By the Chinese Remainder Theorem we
have
$$\begin{array}{ll} \F_2[x]/(x^p+1)=Q & \cong \F_2[x]/(q_0(x))\oplus \F_2[x]/(q_{1}(x))\oplus \ldots
                      \oplus \F_2[x]/(q_{t}(x)) \cong \\
                      & \cong \F_2 \oplus \F_{2^{s(p)}} \oplus \ldots \oplus \F_{2^{s(p)}} \end{array}$$
Furthermore, calling $Q_{j}:=\frac{x^p+1}{q_j(x)}$ we have
$\F_2[x]/(q_j(x)) \cong (Q_j)=:\mathcal{I}_j$ which is a principal
ideal of $\F_2[x]/(x^p+1)$ generated by $Q_j$. Notice that
$Q_j^2=Q_j$ and $Q_iQ_j=0$ if $i\neq j$ (the equalities are $\bmod \
x^p+1$). Then
$$\VV\cong \F_2[x]/(x^p+1)=\mathcal{I}_0\perp \mathcal{I}_1 \perp \ldots \perp \mathcal{I}_t$$
is an orthogonal sum of ideals (generated by orthogonal
idempotents), such that $\mathcal{I}_0\cong \F_2$ and
$\mathcal{I}_1\cong \ldots \cong \mathcal{I}_t\cong \F_{2^{s(p)}}$.
\vspace{2mm}

Let now $\sigma$ be of type $p$-$(c,f)$ and $n=pc+f$. Without lost
of generality we can relabel the coordinates to have
$$\sigma=(1,\ldots,p)(p+1,\ldots,2p)\ldots,((c-1)p+1,\ldots,pc).$$
As $\VV(\sigma)^\perp$ is the set of all even-weight vectors on the
orbits of $\sigma$, we have that
 $v_i=0$, for all $i\in\{pc+1,\ldots,n\}$, for
every $v \in \VV(\sigma)^\perp$. Let us call
$({\VV(\sigma)^\perp})^\ast\leq \F_2^{pc}$ the space obtained
puncturing $\VV(\sigma)^\perp$ on the last $f$ coordinates.\\ We
extend cycle-wise the map $\varphi$ defined in \eqref{eqphi} to a
map $\varphi_p$ as follows
$$\varphi_p:=\underbrace{\varphi\times\ldots\times\varphi}_{c \ \text{times}}:\F_2^{pc}\rightarrow Q^c,$$
via the natural identification $(\F_2^{p})^c=\F_2^{pc}$.\\ Let
$\varphi_p'$ the map $\varphi_p\times \text{id}_f$, where
$\id_f:=\F_2^f\rightarrow \F_2^f$ is the identity map, so that
$\varphi_p':\F_2^{n}\stackrel{\sim}{\rightarrow} Q^c\oplus \F_2^f$.
This map gives an isomorphism of vector spaces
$$\VV=\F_2^n\cong \F_2^{c+f} \oplus \F_{2^{s(p)}}^c\oplus \ldots \oplus \F_{2^{s(p)}}^c.$$
It is easy to observe that $\varphi_p'(\VV(\sigma))\cong \F_2^{c+f}$
and $\varphi_p(({\VV(\sigma)^\perp})^\ast)\cong
\F_{2^{s(p)}}^c\oplus \ldots \oplus \F_{2^{s(p)}}^c$. Furthermore
${\varphi_p'}_{|_{\VV(\sigma)}}=\pi_\sigma$.

\vspace{2mm}

Let us come back to the subcode $\EE(\sigma)$.\\ Clearly, if
$\sss(p)<p-1$, so that $t>1$, $\EE(\sigma)$ can be decomposed
further. A very
nice investigation of this case is contained in \cite{FeulNe}.\\
Here we consider only the fundamental case in which $\sss(p)=p-1$.
Then
$$\pi_\sigma(\CC(\sigma))\leq \F_2^{c+f} \qquad \text{and} \qquad \varphi_p(\EE(\sigma)^\ast)\leq \F_{2^{p-1}}^c,$$
where $\EE(\sigma)^\ast$ is the code obtained puncturing
$\EE(\sigma)$ on the last $f$ coordinates.\\
We conclude this section stating an important theorem, proved by
Vassil I. Yorgov.
\begin{theorem}[\cite{YorgovOddPrimes}]\label{theoremodd} Let $\CC$ be a binary code
with an automorphism $\sigma$ of odd prime order $p$, with
$\sss(p)=p-1$. Then the following are equivalent:
\begin{enumerate}
  \item[\rm a)] $\CC$ is self-dual.
  \item[\rm b)] $\pi_\sigma(\CC(\sigma))$ is self-dual and $\varphi_p(\EE(\sigma)^\ast)$
  is Hermitian self-dual.
\end{enumerate}
\end{theorem}

\begin{remark}
``$\pi_\sigma(\CC(\sigma))$ is self-dual if $\CC$ is self-dual''
holds for \mbox{every odd prime $p$} (see for example
\cite{Conway}). Does it hold also for $p=2$? In general the answer
is ne\-ga\-tive. For example, there are automorphisms of order $2$
of the extended Hamming Code of length $8$ for which it holds true
and others for which it is false.
\end{remark}

\section{Cyclic group of order $2p$ ($p$ an odd prime)}\label{section2p}

Throughout this section we consider $\CC$, a self-dual code of even
length $n$, and $\sigma_{2p}\in\Aut(\CC)$ of order $2p$, where $p$
is an odd prime. We show some module theoretical properties of such
a code, assuming that the involution $\sigma_2:=\sigma_{2p}^p$
acts fixed point freely on the $n$ coordinates. \\
Without loss of generality, we may assume that
$$\sigma_2=\sigma_{2p}^p=(1,2)(3,4)\ldots(n-1,n).$$
We consider the natural projection
$\pi_{\sigma_2}:\CC(\sigma_2)\rightarrow \F_2^{\frac n 2}$ and the
map $$\phi:\CC\rightarrow \F_2^{\frac n 2},$$ with
$(c_1,c_2,\ldots,c_{n-1},c_n)\overset{\phi}{\mapsto}(c_1+c_2,\ldots,c_{n-1}+c_n)$.\\
Stefka Bouyuklieva proved \cite{Baut2} that
$$\phi(\CC)\leq \pi_{\sigma_2}(\CC(\sigma_2))=\phi(\CC)^\perp.$$
In particular,
$$\phi(\CC) = \pi_{\sigma_2}(\CC(\sigma_2)) = \phi(\CC)^\perp \Leftrightarrow \dim \, \pi_{\sigma_2}(\CC(\sigma_2)) = \dim \, \CC(\sigma_2) = \frac{n}{4}.$$
Starting from this easy observation, we proved the following result,
that is the crucial theorem of our joint work with W. Willems.

\begin{theorem}[\cite{BW}]\label{proj}
The code $\CC$ is a projective $\F_2 \langle \sigma_{2p}
\rangle$-module if and only if $\pi_{\sigma_2}(\CC(\sigma_2))$ is a
self-dual code.
\end{theorem}

One of the reasons which makes interesting to determine if the code
is projective is explained in the following remark.

\begin{remark}\label{decompositionxx}
Let $G$ be a finite group and $\mathcal{M}$ a projective
$KG$-submodule. Then for \emph{every} decomposition
$$\soc(\mathcal{M})=V_1\oplus\ldots\oplus V_m$$
of the socle in irreducible $KG$-submodules, we have
$$\mathcal{M}= \PP(V_1)\oplus\ldots\oplus \PP(V_m),$$
where $\PP(V_i)$ is the projective cover of $V_i$ in $\mathcal{M}$,
for all $i\in\{1,\ldots,m\}$.
\end{remark}

So, whenever we have a projective module, there are several
restrictions on its structure and, in particular, the knowledge of
its socle gives us a lot of information about the whole module.

\subsection{Consequences on the structure of $\CC$}

We deduce  some properties of $\CC$ related to the action of the
automorphism $\sigma_{2p}$.

Since $\sigma_2$ acts fixed point freely, $\sigma_{2p}$ is of type
$2p$-$(w,0,x;0)$ for certain $x,w\in\N$ such that $n=2px+2w$. Thus
we have the following decomposition of the $\F_2\langle \sigma_{2p}
\rangle$-module $\F_2^n$:
$$\F_2^n\cong\underbrace{\F_2 \langle
\sigma_{2p} \rangle\oplus \ldots\oplus \F_2 \langle \sigma_{2p}
\rangle}_{x \ \text{times}}\oplus \underbrace{\F_2 \langle \sigma_2
\rangle \oplus \ldots \oplus \F_2 \langle \sigma_2 \rangle}_{w \
\text{times}}.$$ By Section \ref{sectionprime}, recalling that
$\F_2\langle\sigma_{2p}\rangle\cong\F_2\langle\sigma_{2p}^2\rangle\otimes\F_2\langle\sigma_{2p}^p\rangle$
we get
$$\F_2^n\cong
\underbrace{\begin{array}{c} V_0 \\ V_0 \end{array}\oplus \ldots
\oplus \begin{array}{c} V_0 \\ V_0 \end{array}}_{x + w \
\text{times}} \oplus \ldots \oplus \underbrace{\begin{array}{c}
V_{t}
\\ V_{t} \end{array}\oplus \ldots \oplus \begin{array}{c} V_{t}
\\ V_{t} \end{array}}_{x \ \text{times}},
$$
where $t:=\frac{p-1}{\sss(p)}$, $V_0 \cong \F_2$, $V_i$ is an
irreducible module of dimension $\sss(p)$ for every
$i\in\{1,\ldots,t\}$ and $\begin{array}{c} V_j \\ V_j \end{array}$
is a non-split extension of $V_j$ by $V_j$ for every
$j\in\{0,\ldots,t\}$.

Then we get the following result for self-dual codes.

\begin{prop}[\cite{BW}]\label{decomp}
Let $\CC$ be a self-dual binary linear code of length $n$ and
suppose $\sigma_{2p}\in\Aut(\CC)$ of type $2p$-$(w,0,x;0)$. Then the
code $\CC$ has the following structure as an $\F_2\langle
\sigma_{2p} \rangle$-module:
$$\CC=\underbrace{\begin{array}{c} V_0 \\ V_0 \end{array}\oplus \ldots \oplus \begin{array}{c} V_0 \\ V_0 \end{array}}_{y_0 \ \text{times}}\oplus\underbrace{\begin{array}{c}          V_0 \end{array}\oplus \ldots \oplus \begin{array}{c}         V_0 \end{array}}_{z_0 \ \text{times}}\oplus \ldots $$ $$ \mbox{} \qquad \ldots \oplus\underbrace{\begin{array}{c} V_{t} \\ V_{t} \end{array}\oplus \ldots \oplus \begin{array}{c} V_{t} \\ V_{t} \end{array}}_{y_{t} \ \text{times}}
 \oplus \underbrace{\begin{array}{c}     V_{t} \end{array}\oplus \ldots \oplus \begin{array}{c}   V_{t} \end{array}}_{z_{t}},
$$
\noindent where
\begin{itemize}
  \item[\rm a)] $2y_0+z_0=x +w$,
  \item[\rm b$_1$)] $2y_i+z_i=x$ for all $i\in\{1,\ldots,t \}$, if $\sss(p)$ is even,
  \item[\rm b$_2$)] $z_i=z_{2i}$ and $y_{i}+y_{2i} + z_i=x$ for all $i\in\{1,\ldots,t\}$, if $s(p)$ is
  odd.
\end{itemize}
In particular $x\equiv z_1 \equiv \ldots \equiv z_t \bmod 2$, if
$\sss(p)$ is even.
\end{prop}

\noindent This quite technical proposition has a strong consequence
in a particular case.

\begin{coro}[\cite{BW}]\label{coro2}
Let $\CC$ be a self-dual binary linear code of length \mbox{$n\equiv
0 \bmod 4$}. Suppose $\sigma_{2p}\in\Aut(\CC)$ of type
$2p$-$(w,0,x;0)$ with $\sss(p)$ even. If $w$ is odd, then
$$\dim \CC(\sigma_2) = \dim \pi_{\sigma_2}(\CC(\sigma_2)) \geq \frac{n}{4} + \frac{\sss(p)
t }{2}=\frac{n}{4} + \frac{p-1}{2},$$ where
$\sigma_2=\sigma_{2p}^p$.\\
In particular $\pi_{\sigma_2}(\CC(\sigma_2))$ is not self-dual so
that $\CC$ is not a projective $\F_2\langle
\sigma_{2p}\rangle$-module.
\end{coro}

\noindent Other consequences of Proposition \ref{decomp} can be
found in \cite{BW}.

\section{Dihedral group of order $2p$ ($p$ an odd prime)}\label{sectiondihedral}

In this section we consider the structure of a self-dual binary
linear code $\CC$ with a dihedral group as subgroup of $\Aut(\CC)$.
We try to generalize here the main idea used  in \cite{FeulNe} by G.
Nebe and Thomas Feulner to approach the case $D_{10}$ for the
extremal self-dual binary linear code of length $72$. The
assumptions we make are somehow too strong, but they make the
notations simpler and they are sufficient for our purposes.

\noindent Let us now suppose that
\begin{itemize}
  \item $p$ is an odd prime with $\sss(p)=p-1$;
  \item $\CC$ is a self-dual binary linear code of length $n$ ($n$ divisible by $2p$);
  \item $\sigma_p\in \Aut(\CC)$ of
  order $p$ is fixed point free (so that the number of cycles is $c=\frac n p$);
  \item $\sigma_2\in \Aut(\CC)$ of order $2$ is fixed point free;
  \item $\langle\sigma_p\rangle \rtimes \langle \sigma_2\rangle\cong
  D_{2p}$ is a dihedral group
 of order $2p$.
\end{itemize}

\noindent As we have seen in Section \ref{sectionprime},
$\CC=\CC(\sigma_p)\oplus \EE(\sigma_p)$.  The action of the
involution $\sigma_2$ and the results of Theorem \ref{theoremodd}
give strong restrictions on the structure to $\CC$, as we will
prove.

\noindent Without lost of generality we can set
$$\sigma_p:=(1,\ldots,p)(p+1,\ldots,2p)\ldots(n-p+1,\ldots,n)$$
and
$$\sigma_2:=(1,p+1)(2,2p)\ldots(p,p+2)\ldots (n-p,n-p+2).$$

\subsection{Preliminaries}

We need to understand better the structure of the field
$\F_{2^{p-1}}$ in its realization as an ideal $\mathcal{I}$ of
$\F_2[x]/(x^p+1)$, presented in Section \ref{sectionprime}.

\begin{remark} In the following we indicate with $a \bmod b$ the remainder of the division of $a$ by $b$. \\ Furthermore, we
indentify the cosets of $\F_2[x]/(x^p+1)$ with their
representatives.
\end{remark}

Remember that the ideal $\mathcal{I}$ is generated by $(1+x)$. It is
straightforward to observe that
$(x+x^2+\ldots+x^{p-1})\in\mathcal{I}$ is the identity of the
field.\\ Since $\sss(p)=p-1$ we have that
$$(1+x),(1+x)^2,(1+x)^4,\ldots,(1+x)^{2^{p-2}}$$
is an $\F_2$-basis of $\F_{2^{p-1}}$. Furthermore
$$a_0(1+x)+a_1(1+x)^2+\ldots+a_{p-2}(1+x)^{2^{p-2}}=$$$$=(a_0+\ldots+a_{p-2})+a_0 x+ a_1 x^2 +\ldots+a_{p-2} x^{2^{p-2}}.$$
Let $\psi:i\mapsto i+\frac{p-1}{2} \bmod p-1$ and
$\Phi_{2^{\frac{p-1}{2}}}$ the Frobenius automorphism of
$\F_{2^{p-1}}$.
$$\Phi_{2^{\frac{p-1}{2}}}((a_0+\ldots+a_{p-1})+a_0 x+ a_1 x^2 +\ldots+ a_{p-2} x^{2^{p-2}})=$$
$$=(a_0+\ldots+a_{p-1})+a_{\psi^{-1}(0)} x+ a_{\psi^{-1}(1)} x^2 +\ldots+ a_{\psi^{-1}(p-2)} x^{2^{p-2}}.$$
If we identify every polynomial with the ordered vector of $\F_2^p$
of its coefficients, the Frobenius automorphism corresponds to a
permutation of $S_p$. \\Since $[2^{\frac{p-1}{2}}]_p=[-1]_p$, the
permutation
$$\prod_{i=1}^{\frac{p-1}{2}}(2^i \bmod p , 2^{\psi(i)} \bmod
p)$$ is equal to $$(1,p-1)(2,p-2)(3,p-3)\ldots
\left(\frac{p-1}{2},\frac{p+1}{2}\right)$$ so that the Frobenius
automorphism corresponds to the following permutation on the
coefficients of polynomials
$$(2,p)(3,p-1)(4,p-2)\ldots
\left(\frac{p+1}{2},\frac{p+3}{2}\right)$$ that inverts the order of
the last $p-1$ coordinates of the cycle of length $p$.

Let us consider now the direct product of two copies of
$\F_{2^{p-1}}$, so that the coefficients live in $\F_2^{2p}$. The
permutation
$$(1,p+1)(2,2p)(3,2p-1)(4,2p-2)\ldots(p,p+2)\in S_{2p}$$
corresponds to $(\alpha,\beta)\mapsto
\left(\Phi_{2^{\frac{p-1}{2}}}(\beta),\Phi_{2^{\frac{p-1}{2}}}(\alpha)\right)$
over $\F_{2^{p-1}}^2$.

\noindent Let us set
$\overline{\alpha}:=\Phi_{2^{\frac{p-1}{2}}}(\alpha)=\alpha^{2^{\frac{p-1}{2}}}$.

It follows easily that the permutation
$$\sigma_2=(1,p+1)(2,2p)\ldots(p,p+2)\ldots(n-p,n-p+2)$$
acts as follows
$$(\alpha_1,\alpha_2,\ldots,\alpha_{c-1},\alpha_c)\mapsto(\overline{\alpha_2},\overline{\alpha_1},\ldots,\overline{\alpha_{c}},\overline{\alpha_{c-1}})$$
on $\F_{2^{p-1}}^c$ ($c$ even).

\subsection{Main theorem}

We can now state the main result. The notations are those fixed in
the introduction of this section.

\begin{theorem}\label{dihedral}
Let $\CC$ be a self-dual code of length $n$ such that $\langle
\sigma_p\rangle \rtimes \langle \sigma_2\rangle$ is a subgroup of
$\Aut(\CC)$. If $\pi_{\sigma_2}(\CC(\sigma_2))$ is self-dual, then
there exist
\begin{itemize}
  \item $\mathcal{A}\leq \F_2^{\frac n 2}$, which is a self-dual binary linear code,
  \item $\mathcal{B}\subseteq \F_{2^{p-1}}^{\frac c 2}$, which is a
$\F_{2^{\frac{p-1}{2}}}$-linear trace-Hermitian self-dual code,
\end{itemize}
such that
$$\CC=\pi_{\sigma_p}^{-1}(\mathcal{A})\oplus \varphi_p^{-1}\left(\langle \pi^{-1} (\mathcal{B}) \rangle_{\F_{2^{p-1}}}\right)$$
where $\pi_{\sigma_p}$ is the natural projection associated to
$\sigma_p$, $\varphi_p$ is the map defined in \mbox{Section
\ref{sectionprime}} and
$$\pi:=\F_{2^{p-1}}^{c} \rightarrow \F_{2^{p-1}}^{\frac c 2}$$
maps
$(\varepsilon_1,\varepsilon_2,\ldots,\varepsilon_{c-1},\varepsilon_{c})\mapsto(\varepsilon_1,\ldots,\varepsilon_{c-1})$.
\end{theorem}

\begin{proof}
As we have proved in Section \ref{sectionprime},
$$\CC=\CC(\sigma_p)\oplus \EE(\sigma_p).$$
Put $\mathcal{A}:=\pi_{\sigma_p}(\CC(\sigma_p))\leq \F_2^{c+f}$.
This is self-dual by Theorem \ref{theoremodd}.

Let us consider $\varphi_p(\EE(\sigma_p))\leq\F_{2^{p-1}}^c$. This
is an Hermitian self-dual code, again by Theorem \ref{theoremodd}.
As we have just shown the action of $\sigma_2$ on
$\varphi_p(\EE(\sigma_p))$ is the following
$$(\varepsilon_1,\varepsilon_2,\ldots,\varepsilon_{c-1},\varepsilon_{c})^{\sigma_2}=(\overline{\varepsilon_2},\overline{\varepsilon_1},\ldots,\overline{\varepsilon_{c}},\overline{\varepsilon_{c-1}})$$
Note that this action is only $\F_{2^{\frac{p-1}{2}}}$-linear.
Furthermore, the fixed code of $\sigma_2$ is
$$\varphi_p(\EE(\sigma_p))(\sigma_2):=\{(\varepsilon_1,\overline{\varepsilon_1},\ldots,\varepsilon_{\frac{c}{2}},\overline{\varepsilon_{\frac{c}{2}}})\in \varphi_p(\EE(\sigma_p))\}.$$
Put $\mathcal{B}:=\pi(\varphi_p(\EE(\sigma_p))(\sigma_2))$. \\
For
$\gamma , \epsilon \in \mathcal{B} $  the Hermitian inner product of
their preimages in $\varphi_p(\EE(\sigma_p))(\sigma_2)$ is
$$ \sum_{i=1}^{\frac c 2}
(\epsilon_i\overline{\gamma_i}+\overline{\epsilon_i}\gamma_i)$$
which is $0$ since $\varphi_p(\EE(\sigma_p))$ is Hermitian
self-dual. Therefore $\mathcal{B}$ is trace-Hermitian
self-orthogonal. We have
$$\dim_{\F_2}(\mathcal{B}) = \dim _{\F_2} (\varphi_p(\EE(\sigma_p))(\sigma_2)) = \frac{1}{2} \dim _{\F_2}
(\varphi_p(\EE(\sigma_p))) $$ since $\varphi_p(\EE(\sigma_p))$ is a
projective $\F_2\langle\sigma_2 \rangle$-module (since
$\pi_{\sigma_2}(\CC(\sigma_2))$ is self-dual), and so $\mathcal{B}$
is self-dual. \\
Since $\dim_{\F_2}(\mathcal{B})=\dim_{\F_{2^{p-1}}}
(\varphi_p(\EE(\sigma_p)))$, the $\F_{2^{p-1}}$-linear code
$\varphi_p(\EE(\sigma_p))\leq \F_{2^{p-1}}^{c}$ is obtained from
$\mathcal{B}$ as stated.
\end{proof}

\section{Interaction between fixed codes}\label{sectioninter}

In this section we investigate the interaction between fixed codes
of different automorphisms. In particular, we want to give an idea
of what can be said in the case that the automorphism group of a
binary linear code (not necessarily self-dual) contains a subgroup
$H$ that is a semidirect product (abelian or not) of two subgroups,
say $H=A\rtimes B$.

\subsection{Non-abelian semidirect products of two subgroups}

Let us start from the non-abelian case.

Actually, in this case we have an action of $H$ on the normal
subgroup $A$ and in particular on the fixed codes of the
automorphisms belonging to $A$. We restrict our attention to a
particular case. However, this case gives some flavor of what can be
done in general.

{\sc Notation}. For $\tau,\sigma\in S_n$ we denote by $\tau^\sigma$
the conjugate of $\tau$ by $\sigma$.

\noindent Let us start with a basic and trivial lemma.

\begin{lemma}\label{conjfixed}
Let $\CC$ be a linear code of length $n$ and take
$\tau\in\Aut(\CC)$. If $\sigma$ is a permutation of $S_n$ then
$$\tau^\sigma\in\Aut(\CC^\sigma)$$
and
$$\CC(\tau)^\sigma=\CC(\tau^\sigma).$$
\end{lemma}

\begin{proof}
The first assertion is clear. Then, for $c\in \CC$ we have
$$c\in \CC(\tau)^\sigma\Leftrightarrow
c^{\sigma^{-1}}\in\CC(\tau)\Leftrightarrow
c^{\sigma^{-1}\tau}=c^{\sigma^{-1}}\Leftrightarrow
c^{\tau^\sigma}=c\Leftrightarrow c\in \CC(\tau^\sigma),$$ which
proves the second assertion.
\end{proof}

This easy observation suggests a construction for codes with
semidirect automorphism subgroups.

\begin{theorem}\label{semid}
Let $\CC$ be a binary linear code. Suppose that $G=E_m\rtimes H$ is
a subgroup of $\Aut(\CC)$, where $E_m$ is an elementary abelian
$p$-group and $H$ acts transitively on $E_m^\times$. Then
$$\sum_{\varepsilon\in E_m^\times}\CC(\varepsilon)=\sum_{\kappa\in H}\CC(\varepsilon_0)^\kappa$$
for any $\varepsilon_0\in E_m^\times$.
\end{theorem}

\begin{proof}
It follows directly from Lemma \ref{conjfixed}.
\end{proof}

Then we have the following.

\begin{coro}\label{corosemid}
Let $p$ be a Mersenne prime, that is $p=2^r-1$ for a certain
$r\in\N$. Let $E_{2^r}$ be an elementary abelian group of order
$2^r$ and let $G=E_{2^r}\rtimes \langle \sigma_p\rangle$, where
$\sigma_p$ is an
automorphism of order $p$ {\rm(}$G$ non abelian{\rm)}.\\
Suppose that $\CC$ is a binary linear code such that $G$ is a
subgroup of $\Aut(\CC)$. Then for any involution $\varepsilon_0 \in
E_{2^r}$ it holds that
$$\sum_{\varepsilon\in E_{2^r}^\times} \CC(\varepsilon)=\sum_{i=0}^{p-1} \CC(\varepsilon_0)^{\sigma_p^i}.$$
\end{coro}

\begin{proof}
$|E_{2^r}^\times|=2^r-1$. The cyclic group $\langle \sigma_p\rangle$
acts on it. The orbits for this action have order $p$ or order $1$.
Since $p=|E_{2^r}^\times|$ there is only one orbit of order $p$:
supposing the contrary we have $G$ abelian, a contradiction. So the
action is transitive and the assertion follows from Theorem
\ref{semid}.
\end{proof}

Obviously, similar results can be deduced for other groups. Notice
that $A_4$ satisfies the hypothesis of Corollary \ref{corosemid}
with $p=3$.

Let us conclude this subsection, underlining a very useful tool to
investigate further a code with such an automorphism group.\\Let
$\mathcal{D}:=\sum_{\varepsilon\in E_m^\times} \CC(\varepsilon)$.
The group $G$ acts on $\mathcal{Q}:= \mathcal{D}^\perp/\mathcal{D}$
with kernel containing $E_m$. The space $\mathcal{Q}$ is hence a
$\F_2\langle \sigma_p\rangle$-module. On this space we still have a
decomposition in the part fixed by $\sigma_p$ and its complement and
we can repeat arguments totally analogous to the ones in Section
\ref{sectionprime}. This gives again a very restrictive structure.

\subsection{Direct products of cyclic groups}

Let us conclude with a few consi\-derations on the interaction
between fixed codes of different automorphisms in the abelian case.
The results of this subsection can be generalized to any abelian
finite group, but the notation would become too complex.

We consider in particular the group $C_p\times C_q$ with $p,q$ not
necessarily distinct primes. This case gives an idea of what can be
said in a general context.

Let us suppose that $\CC$ is a code (not necessarily self-dual) such
that $C_p\times C_q$ is a subgroup of $\Aut(\CC)$ with $C_p =
\langle\sigma_p\rangle$, $C_q = \langle\sigma_q\rangle$, cyclic
groups of prime (not necessarily distinct) order.\\
 Let $\sigma_p$
be of type $p$-$(c,f)$. Then
$$\pi_{\sigma_p}(\CC(\sigma_p))\leq \F_2^{c+f}.$$
Every element of $\textnormal{C}_{S_n}(\sigma_p)$ (the centralizer
of $\sigma_p$ in $S_n$) acts on the orbits of $\sigma_p$. So we can
define naturally a projection
$$\eta_{\sigma_p}:\textnormal{C}_{S_n}(\sigma_p) \rightarrow S_{c+f}$$
that maps $\tau \in \textnormal{C}_{S_n}(\sigma_p)$ on the
permutation corresponding to the action of $\tau$ on the orbits of
$\sigma_p$. \\
If $\sigma_q$ is of type $q$-$(c',f')$ we can define
in a completely analogous way
$$\eta_{\sigma_q}:\textnormal{C}_{S_n}(\sigma_q) \rightarrow S_{c'+f'}.$$
We collect in the following some observations.

\begin{remark}\label{remarkxxxx}
Let $\CC$ be a code such that $C_p\times C_q \leq \Aut(\CC)$ with
$C_p = \langle\sigma_p\rangle$, $C_q = \langle\sigma_q\rangle$,
cyclic groups of prime (not necessarily distinct) order. Then
\begin{itemize}
\item[\rm a)] $\eta_{\sigma_p}(\sigma_q) \in
\Aut(\pi_{\sigma_p}(\CC(\sigma_p)));$
\item[\rm b)] $\eta_{\sigma_q}(\sigma_p) \in
\Aut(\pi_{\sigma_q}(\CC(\sigma_q)));$
\item[\rm c)] $\eta_{\eta_{\sigma_p}(\sigma_q)}(\pi_{\sigma_p}(\CC(\sigma_p)) (\eta_{\sigma_p}(\sigma_q)))=\eta_{\eta_{\sigma_q}(\sigma_p)} (\pi_{\sigma_q}(\CC(\sigma_q)) (\eta_{\sigma_q}(\sigma_p)));$
\item[\rm d)] if $p,q$ are distinct and $\sigma_p\sigma_q$ is of type
$pq$-$(a,b,c;f)$ then $\eta_{\sigma_p}(\sigma_q)$ is of type
$q$-$(c+b,a+f)$ and $\eta_{\sigma_q}(\sigma_p)$ is of type
$p$-$(c+a,b+f)$.
\end{itemize}
\end{remark}

\noindent Notice that a) and b) are strong conditions on the fixed
codes.

\section{The automorphism group of an extremal self-dual code of
\mbox{length $72$}}

The existence of an extremal self-dual code of length $72$ is a
long-standing open problem of classical Coding Theory \cite{S}.

We give here a brief overview of the investigation of its possible
automorphism groups. We do not follow a chronological order, nor we
mention all the papers related to the topic. Our aim is to outline
all the steps necessary to prove the final theorem and to underline
where the methods presented in the previous sections can be applied.

For all this section let $\CC$ be an extremal self-dual $[72,36,16]$
code.

\subsection{Cycle-structure of the automorphisms}

In order to get information on the whole group $\Aut(\CC)$, we begin
to investigate the cycle-structure of the possible automorphisms.

John H. Conway and Vera Pless, in a paper submitted in 1979
\cite{Conway}, were the first who faced this problem. In particular
they focused on the possible automorphisms of odd prime order. They
proved that
\begin{itemize}
  \item \emph{only $9$ types of automorphism of odd prime order may occur in $\Aut(\CC)$,
  namely \
$23$-$(3,3)$, \ $17$-$(4,4)$, \ $11$-$(6,6)$, \ $7$-$(10,2)$, \
$5$-$(14,2)$, \ $3$-$(18,18)$, $3$-$(20,12)$, $3$-$(22,6)$ and
$3$-$(24,0)$.}
\end{itemize}
They used arguments based on combinatorial properties of the codes.

Between 1981 and 1987, V. Pless, John G. Thompson, W. Cary Huffman
and V.I. Yorgov \cite{Pless23,Pless17,Huffman11} proved that
\begin{itemize}
  \item \emph{automorphisms of orders $23$, $17$ and $11$ cannot occur in
$\Aut(\CC)$.}
\end{itemize}
Between 2002 and 2004, S. Bouyuklieva \cite{Bord2,Bord3} proved that
\begin{itemize}
  \item \emph{the eventual elements of order $2$ and $3$ in $\Aut(\CC)$ are fixed point free.}
\end{itemize}
More recently T. Feulner and G. Nebe \cite{FeulNe} showed that also
\begin{itemize}
  \item \emph{automorphisms of orders $7$ cannot occur in $\Aut(\CC)$.}
\end{itemize}
The techniques used are different case by case, but the main tool is
the decomposition of codes with an automorphism of odd prime order
discussed in Section \ref{sectionprime}. Let us summarized these
results in the following.

\begin{prop}\label{cyclestructure}
Let $\sigma$ be an automorphism of prime order of a self-dual
$[72,36,16]$ code. Then $\sigma$ can be only of the following types:
$2$-$(36,0)$, $3$-$(24,0)$ and $5$-$(14,2)$.
\end{prop}

An immediate consequence of Proposition \ref{cyclestructure} is that
$\Aut(\CC)$ does not contain elements of order $15$, $16$, $25$ and
$27$. Furthermore, the possible composite orders are $4$, $6$, $8$,
$9$, $10$, $12$, $18$, $36$ and $72$.\\ G. Nebe, Nikolay Yankov and
the author \cite{Neven,Yankov9,Baut6}, excluded orders $10$, $9$ and
$6$, respectively. Order $10$ can be excluded just looking at the
automorphism groups of self-dual $[36,18,8]$ codes, classified in
\cite{Gaborit}, which are the projection of possible fixed codes of
involutions, and using Remark \ref{remarkxxxx}.  The methods used
for order $9$ are a refinement of those in Section
\ref{sectionprime}. For order $6$ we used strongly the results
contained in Section \ref{section2p}.

Finally, very recently V.I. Yorgov and Daniel Yorgov proved that
automorphism of order $4$ are not possible \cite{YY}. So we have the
following.

\begin{prop}\label{remarkprime}
Let $\sigma$ be a non-trivial automorphism of a self-dual
$[72,36,16]$ code. Then its order is a prime among $\{2,3,5\}$.
\end{prop}

\subsection{Structure of the whole group}

Once we have information on the cycle-structure of the
automorphisms, we can investigate the structure of the whole
group.\\ By Proposition \ref{cyclestructure} we have immediately
that
$$|\Aut(\CC)|=2^a3^b5^c$$
where $a,b,c$ are nonnegative integers.

S. Bouyuklieva was the first, in 2004, who studied the order of
$\Aut(\CC)$. She proved \cite{Bord3} that
\begin{itemize}
  \item \emph{$25$ does not divide $|\Aut(\CC)|$.}
\end{itemize}
This means that $|\Aut(\CC)|=2^a3^b5^c$
with $a,b$ nonnegative integers and $c=0,1$. \\
If $c=1$ then
\begin{itemize}
  \item if $\sigma\in\Aut(\CC)$ has order $5$, $|\textnormal{N}_{\Aut(\CC)}(\sigma)|=2^d 5$, with $d=0,1$ \cite{Yorgov}.
  \item $|\{ \text{aut. of order} \ 5 \ \text{in} \ \Aut(\CC)\}|=4\cdot \frac{|\Aut(\CC)|}{2^\delta 5}$.
\end{itemize}
So, by Burnside Lemma,  $$\frac{1}{|\Aut(\CC)|}\left(72+\gamma \cdot
2\cdot
  \frac{4\cdot |\Aut(\CC)|}{2^\gamma5}\right)= \frac{72}{2^\alpha3^\beta5^\gamma}+\gamma\cdot \frac{8}{2^\delta 5} \ \ \in \N$$
$$\Downarrow$$
$$|\Aut(\CC)|\in \{ 1, 2, 3, 4, 5, 6, 8, 9, 10, 12, 18, 24, 30, 36,
60, 72, 180, 360 \} \ \ (\star).$$ By Proposition \ref{remarkprime},
we have that
\begin{itemize}
  \item $\Aut(\CC)$ is trivial or isomorphic to one of the following: $C_2$, $C_3$, $C_2\times C_2$, $C_5$, $S_3$, $C_2\times
  C_2\times C_2$, $C_3\times C_3$, $D_{10}$, $A_4$, $(C_3\times
  C_3)\rtimes C_2$ (the generalized dihedral group of order $18$) or
  $A_5$,
\end{itemize}
since all other groups of order in $(\star)$  have elements of
composite order (for a library of Small Groups see for example
\cite{SmallGroups}).\\
T. Feulner and G. Nebe \cite{FeulNe} proved that
\begin{itemize}
\item $\Aut(\CC)$ does not contain a subgroup isomorphic to $C_3\times C_3$ or $
D_{10}$.
\end{itemize}
The author, in a joint paper \cite{BDN} with F. Dalla Volta and G.
Nebe, proved that
\begin{itemize}
\item $\Aut(\CC)$ does not contain a subgroup isomorphic to
$S_3$ or $A_4$.
\end{itemize}
Finally, the author proved \cite{BE8} that
\begin{itemize}
\item $\Aut(\CC)$ does not contain a subgroup isomorphic to
$C_2\times C_2\times C_2$.
\end{itemize}
The methods used for $C_3\times C_3$ are a refinement of those
presented in Section \ref{sectionprime}. The cases of $D_{10}$ and
$S_3$ involve the methods of Section \ref{sectiondihedral}. For
$A_4$ and $C_2\times C_2\times C_2$ we applied the methods of
Section \ref{sectioninter} with some more particular observations.

Let us summarize all these results in a theorem.

\begin{theorem}
Let $\CC$ be self-dual $[72,36,16]$ code. Then $\Aut(\CC)$ is
trivial or isomorphic to $C_2$, $C_3$, $C_2\times C_2$ or $C_5$.
\end{theorem}

\begin{remark}
The possible automorphism groups of a putative extremal self-dual
code of length $72$ are abelian and very small. So this code is
almost a rigid object (i.e. without symmetries) and it might be very
difficult to find it, if it exists.
\end{remark}

\bibliographystyle{amsplain}

\end{document}